\documentclass[11pt,a4paper]{article}
\usepackage[T1]{fontenc}
\usepackage[bitstream-charter]{mathdesign}
\usepackage[latin1]{inputenc}							% Input encoding
\usepackage{amsmath}									% Math
\usepackage{amsthm}
\usepackage{amsbsy}

\usepackage{esint}
\usepackage{graphicx}

\usepackage{xcolor}
\definecolor{bl}{rgb}{0.0,0.2,0.6} 

\usepackage{comment}

\usepackage{sectsty}
\usepackage[compact]{titlesec} 
\allsectionsfont{\color{bl}\scshape\selectfont}

\newcommand{\bxi}{\ensuremath{\boldsymbol{\xi}}}
\newcommand{\Bi}{\ensuremath{\mathrm{Bi}}}

\numberwithin{equation}{section}
\numberwithin{figure}{section}
\theoremstyle{plain}
\newtheorem{thm}{\theoremname}
\theoremstyle{definition}

\theoremstyle{remark}
\newtheorem{rem}[thm]{\remarkname}
\theoremstyle{plain}
\newtheorem{lem}[thm]{\lemmaname}
\newtheorem{prop}[thm]{\propname}

\newtheorem{cor}[thm]{\corname}

\providecommand{\lemmaname}{Lemma}
\providecommand{\exmpname}{Example}
  \providecommand{\remarkname}{Remark}
  \providecommand{\examplename}{Example}
\providecommand{\theoremname}{Theorem}
\providecommand{\propname}{Proposition}
\providecommand{\corname}{Corollary}
\providecommand{\conj}{Conjecture}

\renewcommand*{\P}{\operatorname{P}}
\newcommand*{\E}{\mathrm{E}}
\newcommand*{\Var}{\operatorname{Var}}

\newcommand*{\DR}{\mathbb{R}}

\makeatletter							% Begin definition
\def\printtitle{%						% Define command: \printtitle
    {\color{bl} \centering \huge \sc \textbf{\@title}\par}}		% Typesetting
\makeatother							% End definition

\title{Multilevel path simulation for  weak approximation schemes 
\vspace{10pt}
}

\makeatletter							% Begin definition
\def\printauthor{%					% Define command: \printauthor
    {\centering \small \@author}}				% Typesetting
\makeatother							% End definition

\author{%
	Denis Belomestny\footnote{\footnotesize This research was partially supported by the Deutsche
      Forschungsgemeinschaft through the SPP 1324 ``Mathematical methods for extracting quantifiable     information from complex systems'' and   by
Laboratory for Structural Methods of Data Analysis in Predictive Modeling, MIPT, RF government grant, ag. 11.G34.31.0073.} and Tigran Nagapetyan \\
	Duisburg-Essen University, National Research University Higher School of Economics,
 Weierstrass Institute for Applied Analysis and Stochastics\\
	\vspace{20pt}
	}
	
\usepackage{fancyhdr}
\usepackage{lastpage}	
\usepackage[runin]{abstract}			% runin option for a run-in title
\abslabeldelim{\quad}						% 
\begin{document}
%%% Top of the page: Author, Title and Abstact
\printtitle

\printauthor
\begin{abstract}
In this paper we discuss the possibility of using multilevel Monte Carlo  (MLMC)  methods for weak approximation schemes.
It turns out that by means of a simple coupling between consecutive time discretisation levels, one can achieve the same complexity gain as under the presence of  a strong convergence. We exemplify this general idea in the case of  weak Euler scheme for L\'evy driven stochastic differential equations, and show that, given a weak convergence of order \(\alpha\geq 1/2,\)   the complexity of the corresponding ``weak'' MLMC estimate   is of order \(\varepsilon^{-2}\log ^{2}(\varepsilon).\)  The numerical performance of the new ``weak''    
MLMC method is illustrated by several numerical examples. 
\end{abstract}

%%% Start of the 'real' content of the article, using a two column layout
\section{Introduction}
The multilevel path simulation method introduced in Giles \cite{giles2008multilevel} has gained huge  popularity as a complexity reduction tool in 
recent times. The main advantage of  the MLMC methodology is that it can be simply applied to various situations and 
requires almost no prior knowledge on the path generating process. Any multilevel Monte Carlo (MLMC) algorithm uses a number of levels of resolution, \(l=0,1,\ldots, L,\) with \(l=0\) being the coarsest, and \(l=L\) being the finest. In the context of a SDE simulation on the interval \([0,T]\), level \(0\) corresponds to one timestep \(\Delta_0=T,\) whereas the level \(L\)  has \(2^L\) uniform timesteps \(\Delta_L=2^{-L}T.\) 
\par
Assume that a filtered probability space \((\Omega,\mathcal{F},\P,(\mathcal{F}_t))\) is given.    Consider now a \(d\)-dimensional process \((X_t)\) solving the following L\'evy driven SDE
\begin{eqnarray}
\label{sde}
X_t=X_0+\int_0^t a(X_{s-})\,dL_s,
\end{eqnarray}
where  \(X_0\) is a \(\mathbb{R}^d\)-valued random variable, \(L_t=(L^1_t,\ldots,L^m_t),\) \(t\geq 0\)  is a  \(m\)-dimensional L\'evy process and  the mapping \(a:\) \(\DR^d\times \DR^m\mapsto\DR^d\)  is Lipschitz continuous and has at most linear growth on \([0,T]\) so that the solution of \eqref{sde} is well defined.
Our aim is to estimate the expectation \(\E[f(X_T)],\) where \(f\) is a Lipschitz continuous function from \(\DR^d\) to \(\DR.\) Let \(X^l_T\) be an approximation for \(X_T\) by means of a numerical discretisation with time step \(\Delta_l\) (for various discretisation methods for \eqref{sde} see, e. g. Platen and Bruti-Liberati~\cite{platen2010numerical} or the recent review of Jourdain and Kohatsu-Higa \cite{jourdain2011review}). 
The main idea of the multilevel approach pioneered in Giles \cite{giles2008multilevel}  consists in writing the expectation of the finest approximation \(\E[f(X^L_T)]\) as a telescopic sum
\begin{eqnarray*}
\E[f(X^L_T)]=\E[f(X^0_T)]+\sum_{l=1}^L \E[f(X^l_T)-f(X^{l-1}_T)]
\end{eqnarray*}
and then applying  Monte Carlo  to estimate each expectation in the above telescopic sum.
One important prerequisite for MLMC to work is that \(X^l_T\) and \(X^{l-1}_T\) are coupled in some way and this can be achieved by using the same discretised trajectories of the underlying L\'evy processes to construct  the consecutive approximations \(X^l_T\) and \(X^{l-1}_T.\) The degree of coupling is usually measured in terms of the variance \(\Var[f(X^l_T)-f(X^{l-1}_T)]\). 
It is shown in Giles \cite{giles2008multilevel} (see also Giles and Xia \cite{giles2014multilevel}), that under the assumptions 
\begin{eqnarray}
\label{ml_ass}
\bigl|\E[f(X^L_T)]-\E[f(X_T)]\bigr|\leq c_1\Delta_L^{\alpha},\quad \Var\left[f(X^l_T)-f(X^{l-1}_T)\right]\leq c_2 \Delta_l^{\beta}, 
\end{eqnarray}
with some \(\alpha\geq 1/2,\) \(\beta>0,\) \(c_1>0\) and \(c_2>0,\) the computational complexity of the resulting multilevel estimate needed to achieve the accuracy \(\varepsilon\) (in terms of RMSE)  is proportional to 
\begin{eqnarray*}
\mathcal{C}=
\begin{cases}
\varepsilon^{-2}, & \beta>1, \\
\varepsilon^{-2}\log^2(\varepsilon), & \beta=1, \\
\varepsilon^{-2-(1-\beta)/\alpha}, & 0<\beta <1.
\end{cases}
\end{eqnarray*}
The standard way of checking the assumptions \eqref{ml_ass} is to prove that the underlying approximation scheme has   weak convergence of order \(\alpha\) and strong convergence of order \(\beta/2.\) Indeed, in the latter case  we have for any Lipschitz continuous function \(f,\)
\begin{align*}
\Var\left[f(X^l_T)-f(X^{l-1}_T)\right]&\leq c_f\E\left[\bigl|X^l_T-X_T\bigr|^2\right] +c_f\E\left[\bigl|X^{l-1}_T-X_T\bigr|^2\right]\\ &\leq 2c_f\Delta_l^{\beta}
\end{align*}
with some constant \(c_f>0\) depending on \(f.\) 
However, in recent years the so-called weak approximation schemes, i.e., schemes that, in general, fulfil only the first assumption in \eqref{ml_ass} became quite popular. 
The weak Euler scheme is a first-order scheme with \(\alpha=1\), and has been studied by many researchers. Talay and Tubaro \cite{talay1990expansion} show the first-order convergence of the weak Euler scheme. The fact that the convergence rate of the Euler scheme also holds for certain irregular functions  under a H\"ormander type condition has been proved by Bally and Talay \cite{bally1995euler} using Malliavin calculus.  The It\^o-Taylor (weak-Taylor) high-order scheme is a natural extension of the weak Euler scheme. In the continuous diffusion case, some new discretization schemes (also called Kusuoka type schemes) which are of order \(\alpha\geq 2\) without the Romberg extrapolation have been introduced by Kusuoka \cite{kusuoka2004approximation}, Lyons and Victoir \cite{lyons2004cubature}, Ninomiya and Victoir \cite{ninomiya2008weak}, and Ninomiya and Ninomiya \cite{ninomiya2009new}. A general class of weak approximation methods, comprising  many well known discretisation schemes, was constructed in Kohatsu-Higa and Tanaka \cite{tanaka2009operator}.
The main advantage of the weak approximation schemes is that simple discrete random variables can be used instead of the L\'evy increments. Unfortunately, due to the absence of the strong convergence, the MLMC methodology can not be  directly used with the weak approximation schemes.
In this paper we make an attempt to overcome this difficulty and develop a kind of ``weak'' MLMC approach which can be applied to various weak approximation schemes.
\par
The plan of the paper is as follows. First,  we recall the  Euler scheme for \eqref{sde} and discuss its convergence properties. Next we show how to construct the corresponding  MLMC algorithm, which is able to reduce the complexity of the standard MC to  order \(\varepsilon^{-2}\log^2(\varepsilon)\) under only requirement that the Euler scheme converges weakly.  Finally, we analyse the numerical performance of the presented weak MLMC algorithms.
\section{Euler scheme for L\'evy driven SDE}
Fix some \(n\in \mathbb{N}\) and set \(\Delta=T/n.\)  Denote \(\Delta L_j=L_{j\Delta}-L_{(j-1)\Delta},\) \(j=1,\ldots, n.\)  For a fixed random vector \(X_0, \) the  Euler scheme  for \eqref{sde} reads as follows
\begin{eqnarray}
\label{euler_scheme}
X^{\Delta}_0&=&X_0,
\\
\nonumber
X^{\Delta}_{j\Delta}&=&X^{\Delta}_{(j-1)\Delta}+a\bigl(X^{\Delta}_{(j-1)\Delta}\bigr)\,\Delta L_j, \quad j=1,\ldots, n.
\end{eqnarray}
The convergence of the scheme \eqref{euler_scheme} was extensively studied in the literature. The first convergence result is due to Talay and Tubaro \cite{talay1990expansion}, who proved that in the case of a diffusion processes with \(L\) being a Brownian motion plus drift, the scheme weakly converges  with order \(1.\) In the case of the general L\'evy processes, the convergence of \eqref{euler_scheme} was studied in Protter and Talay \cite{protter1997euler}, where it is shown that, under some  assumption on the function \(a\) and the driving L\'evy process \(L\), the weak convergence rate \(1/n\) can  be recovered.  In fact, the main drawback of the scheme \eqref{euler_scheme} is the necessity to sample from the distribution of \(\Delta L_j\) exactly. Although such exact sampling can be possible for particular L\'evy processes (see \cite{protter1997euler} for some examples), in general this turns out to be a hard numerical problem. This is why Jacod et al \cite{jacod2005approximate} proposed to replace the increments \(\Delta L_j\) of the original L\'evy process  by simple random vectors \(\zeta_j\) which are easy to simulate.  It is shown in \cite{jacod2005approximate} that if the distributions of \(\Delta L_j\) and \(\zeta_j\) are sufficiently close, then the weak convergence rate \(1/n\) continues to hold.
These results on weak convergence should be compared with ones on pathwise or strong convergence.
In fact, the strong convergence rates usually depend on the characteristics of the L\'evy process \(L.\) 
For example, Rubenthaler \cite{rubenthaler2003numerical} studied the strong error when neglecting small jumps. He obtains the estimate of the form
\begin{eqnarray*}
\E\left[\max_{j=0,\ldots, n}|X^{\Delta}_{j\Delta}-X_{j\Delta}|^2\right]\lesssim \left(n^{-1}+\int_{|z|\leq \epsilon} z^2\nu(dz)\right)
\end{eqnarray*}
with \(\nu\) being the L\'evy measure of \(L.\) So the rates become quite poor if \(\nu\) diverges at zero like \(z^{-\alpha}\)
with \(\alpha\) close to \(2.\) Recently, Fournier \cite{fournier2011simulation} has proposed a coupling method which allows to get better rates of pathwise convergence in a one-dimensional case. He constructed an approximation \(X^{n,\epsilon},\) satisfying
 \begin{eqnarray*}
\E\left[\max_{j=0,\ldots, n}|X^{{n,\epsilon}}_{j\Delta}-X_{j\Delta}|^2\right]\lesssim \left(n^{-1}+n\,\frac{m_{4,\epsilon}(\nu)}{m_{2,\epsilon}(\nu)}\right)
\end{eqnarray*}
with \(m_{k,\epsilon}(\nu)=\int_{|z|\leq \epsilon} |z|^k\nu(dz).\) The approximation \(X^{n,\epsilon}\) is constructed by replacing the jumps of \(L\) smaller than \(\epsilon\) by an independent Brownian motion.   In order to  prove  a bound for the Wasserstein distance between \(X\) and \(X^{n,\epsilon},\)   a suitable coupling was used. Note that since \(X\) is unknown, such coupling is not implementable. A similar coupling idea in the multidimensional setting was used in Dereich \cite{dereich2011multilevel} to design a multilevel path simulation approach for \eqref{sde}.   
\section{Multilevel path simulation for weak Euler scheme}
In order to  successfully apply the multilevel approach, one needs to ensure that \eqref{ml_ass} hold. If the scheme \eqref{euler_scheme}
has strong convergence of order \(\beta/2,\) i.e.,
\begin{eqnarray*}
\E\left[\max_{j=0,\ldots, n}|X^{\Delta}_{j\Delta}-X_{j\Delta}|^2\right]\lesssim \Delta^{\beta},
\end{eqnarray*}
then the conditions \eqref{ml_ass}  hold with \(\alpha=\beta/2.\) However, if some approximations \(\zeta_j,\) \(j=1,\ldots, n,\) are used instead of the genuine increments \(\Delta L_j,\)  strong convergence is not any longer guaranteed.
Here we propose a general approach how to couple two consecutive approximations of \(X\) in order to guarantee  that the second condition in  \eqref{ml_ass} still holds with \(\beta=1.\) In fact, this would lead to a complexity estimate   \(\varepsilon^{-2}\log^2(\varepsilon),\) does not matter how small is \(\alpha\geq 1/2.\)
\subsection{Coupling idea}
\label{coupling-section}
Let us fix two natural numbers \(n_c\) (``coarse'' discretisation level) and \(n_f\) (``fine'' discretisation level) with \(n_f=2\cdot n_c\) and set \(\Delta_c=T/n_c,\) \(\Delta_f=T/n_f.\) 
In order to couple the Euler approximations \(X^{\Delta_c}\) and \(X^{\Delta_f},\) we are going to couple the random matrices \(\boldsymbol{\zeta}^c\doteq (\zeta^c_{1},\ldots, \zeta^c_{n_c})\in \DR^{n_c}\otimes\DR^{m} \) and \(\boldsymbol{\zeta}^f\doteq (\zeta^f_{1},\ldots, \zeta^f_{n_f})\in \DR^{n_f}\otimes\DR^{m}.\) We define the approximation \(\boldsymbol{\zeta}^c\) for the increments on the coarse level in such a way that the differences
\begin{eqnarray}
\label{coupling}
\zeta^c_{j}-\zeta^f_{2j-1}-\zeta^f_{2j}, \quad j=1,\ldots, n_c.
\end{eqnarray}
are small. In particular, we can take \(\zeta^c_{j}=\zeta^f_{2j-1}+\zeta^f_{2j}.\)
The idea behind this coupling is very simple: in the case of the genuine L\'evy increments we would get 
\[
\Delta^f L_{2j-1}+\Delta^f L_{2j}=L_{2j\Delta}-L_{2(j-1)\Delta}=\Delta^c L_{j}.
\]
Suppose that \(\zeta^f_{1},\ldots, \zeta^f_{n_f} \) are i.i.d. random vectors with moments \(m_{f,1}\doteq\bigl\|\E[\zeta^f_{j}]\bigr\|\) and \(m_{f,2}\doteq\E\bigl[\|\zeta^f_{j}\|^2\bigr].\) The following proposition holds.
\begin{prop}
\label{main_prop}
Suppose that the coefficient function \(a\) in \eqref{sde} is uniformly Lipschitz and has at most linear growth, i.e.,
\begin{eqnarray}
\label{lip}
\|a(x)-a(x')\|\leq L_a\,\|x-x'\|,\quad 
\|a(x)\|^2\leq B^2_a \, (1+\|x\|^2)
\end{eqnarray}
for any \(x,x'\in \DR^d\)  and some positive constants \(L_a\) and \(B_a.\) Denote \(\mathcal{R}_j\doteq\zeta^c_{j}-\zeta^f_{2j-1}-\zeta^f_{2j}\) and suppose that \(\mathcal{R}_j,\) \(j=1,\ldots, n_c,\) are zero mean i. i. d. random vectors. Moreover, assume that \(\E[\|X_0\|^2]<\infty,\) then the following estimate holds
\begin{multline}
\label{main_estimate}
\E\left[\max_{j=0,\ldots, n_c}\bigl|X^{\Delta_f}_{j\Delta}-X^{\Delta_c}_{j\Delta}\bigr|^2\right]\leq  c_1 \left(n_{f} m_{f,2}^{2}+n^2_f m_{f,1}^{2}m_{f,2}+n_f\E\bigl[\|\mathcal{R}_1\|^2\bigr]\right)
\\
\times \exp\left[c_2\left(n_{f}m_{f,2}+n_{f}^{2}m_{f,1}^{2}\right)\right]
\end{multline}
for some constants \(c_1>0,\)  \(c_2>0\) depending on \(L_a\) and \(B_a.\) 
\end{prop}
\begin{cor}
	\label{cor_main}
If \(m_{f,2}=O(\Delta_f),\)  \(m_{f,1}=O(\Delta_f)\) and \(\E\bigl[\|\mathcal{R}_1\|^2\bigr]=O(\Delta^2_f)\) for \(\Delta_f\to 0,\) then
\begin{eqnarray*}
\E\left[\max_{j=0,\ldots, n_c}\Bigl|X^{\Delta_f}_{j\Delta}-X^{\Delta_c}_{j\Delta}\Bigr|^2\right]=O(\Delta^f), \quad \Delta^f\to 0.
\end{eqnarray*}
\end{cor}
\paragraph{Discussion}
First note that the conditions for \eqref{main_estimate} to hold are formulated not in terms of the original increments \((\Delta L_j),\) but rather in terms of their approximations  \((\zeta_j).\) For the case of the exact increments, we obviously have \(m_{f,2}=O(\Delta_f)\) and \(m_{f,1}=O(\Delta_f),\) provided 
\begin{eqnarray*}
\int_{\mathbb{R}^d} \|z\|^2\,\nu(dz)<\infty,
\end{eqnarray*}
where \(\nu\) is a L\'evy measure of \(L.\) Furthermore, observe that  under the assumptions of Corollary~\ref{cor_main}, the second condition in \eqref{ml_ass} holds with \(\beta=1\) independently of the  strong  convergence  order for the corresponding Euler scheme. Finally, let us stress that the assumptions on the coefficient function \(a\) are quite weak and standard in the framework of L\'evy driven SDEs. In fact, they are needed to guarantee existence  and  uniqueness of the solution of \eqref{sde} (see, e.g.,  Ikeda and Watanabe \cite{watanabe1981stochastic}).
\subsection{MLMC algorithm}
\label{mlmc_alg}
Fix some $L>0$ and set $\Delta_{l}=2^{-l}T,$ $l=0,\ldots,L.$ 
Denote 
\begin{eqnarray*}
\boldsymbol{\zeta}_{L}&\doteq&\bigl(\zeta_{L,1},\ldots,\zeta_{L,2^{L}}\bigr)\in\mathbb{R}^{2^{L}}\otimes \mathbb{R}^{m},
\end{eqnarray*}
where the columns of the matrix  $\boldsymbol{\zeta}_{L}$ are i.i.d.  random vectors in \(\DR^m.\) Now we define recursively the independent random matrices  $\boldsymbol{\zeta}_{L-1},\ldots,\boldsymbol{\zeta}_{0}$ with \(\boldsymbol{\zeta}_{l}\in\mathbb{R}^{2^{l}}\otimes \mathbb{R}^{m}\)
via \(\boldsymbol{\zeta}_{l-1}\sim\boldsymbol{\varsigma}(\boldsymbol{\zeta}_{l})=(\varsigma_1(\boldsymbol{\zeta}_{l}),\ldots, \varsigma_{l-1}(\boldsymbol{\zeta}_{l})), \) where each vector \(\varsigma_j(\boldsymbol{\zeta}_{l})\) is coupled with \(\zeta_{l,2j-1}\) and \(\zeta_{l,2j}\) in such a way that all differences 
\begin{eqnarray*}
\varsigma_j(\boldsymbol{\zeta}_{l})-
\zeta_{l,2j-1}-\zeta_{l,2j}, \quad j=1,\ldots, 2^{l-1},
\end{eqnarray*}
are small. For example, one can simply put
\begin{eqnarray}
\label{simple_coupling}
\varsigma_j(\boldsymbol{\zeta}_{l})=
\zeta_{l,2j-1}+\zeta_{l,2j}, \quad j=1,\ldots, 2^{l-1}.
\end{eqnarray}
Next, for any \(l=1,\ldots, L,\) and any random matrix \(\boldsymbol{\zeta}\in \mathbb{R}^{2^{l}}\otimes \mathbb{R}^{m},\)  consider the approximations 
\begin{align*}
&X^l_0(\boldsymbol{\zeta})=X_0,
\\
&X_{j\Delta_{l}}^{l}(\boldsymbol{\zeta})=X_{(j-1)\Delta_{l}}^{l}(\boldsymbol{\zeta})+a\bigl(X_{(j-1)\Delta_{l}}^{l}(\boldsymbol{\zeta})\bigr)\,\zeta_{j}
\end{align*}
with $j=1,\ldots,2^{l},$ and some r. v. \(X_0\in \DR^d.\)   Finally, fix a vector of natural numbers \(\mathbf{N}=(N_0,\ldots,N_L)\) and define a weak MLMC estimate for \(\E[f(X_T)]\) as follows 
\begin{align*}
Y_{L,\mathbf{N}}&\doteq\frac{1}{N_{0}}\sum_{n=1}^{N_{0}}\left[f(X_{T}^{0}(\boldsymbol{\zeta}^{(n)})\right]+
\sum_{l=1}^{L}\frac{1}{N_{l}}\sum_{n=1}^{N_{l}}\left[f\bigl(X_{T}^{l}(\boldsymbol{\zeta}_l^{(n)})\bigr)-f\bigl(X_{T}^{l-1}\bigl(\overline{\boldsymbol{\zeta}}_l^{(n)}\bigr)\bigr)\right],
\end{align*}
where  \(\overline{\boldsymbol{\zeta}}_l^{(n)}=\boldsymbol{\varsigma(\boldsymbol{\zeta}_l^{(n)})}\) and \(\boldsymbol{\zeta}_l^{(n)},\) \(n=1,\ldots, N_l,\) are i.i.d. copies of \(\boldsymbol{\zeta}_l\).

\begin{prop}
\label{prop_euler_compl}
Suppose that the the function \(f\) is Lipschitz continuous and that the distribution of \(\boldsymbol{\zeta}_{L}\) is chosen in such a way that 
\begin{eqnarray}
\label{weak_compl}
\bigl |\E[f(X_{T}^{L}(\boldsymbol{\zeta}_L))]-\E[f(X_{T})]\bigr |\leq c \Delta_L^{\alpha}
\end{eqnarray}
for some \(c>0\) and \(\alpha\geq 1/2.\) Then  under the assumptions of Proposition~\ref{main_prop} and Corollary~\ref{cor_main}, 
and under a proper choice of \(\mathbf{N}\) and \(L,\) the complexity of the estimate \(Y_{L,\mathbf{N}}\)
needed to achieve the accuracy \(\varepsilon\) (as measured by RMSE) is of order \(\varepsilon^{-2}\log^2(\varepsilon).\)
\end{prop}
\begin{rem}
The distribution of the matrix \(\boldsymbol{\zeta}_{l}\) under coupling \eqref{simple_coupling}, changes with \(l\) in a rather simple way and can be  found explicitly  in many interesting cases (see examples below). In general, one can compute the characteristic function of each vector \(\zeta_{l,j}\) in a closed form, provided the characteristic function of \(\zeta_{L,j}\) is known explicitly. Using the Fourier inversion formula, one can then compute the density of each \(\zeta_{l,j}.\) Let us also note  that there is a  lot of freedom in the choice of the finest approximation \(\boldsymbol{\zeta}_{L}\) satisfying   \eqref{weak_compl}.
\end{rem}
\section{Examples}
\label{examples}
\subsection{Diffusion processes}
\label{diffusion}
Consider now a \(d\)-dimensional diffusion process \((X_t)\) solving the SDE
\begin{eqnarray}
\label{sde:diffusion}
X_t=b(X_t)\,dt+\sigma(X_t)\,dW_t,\quad t\in [0,T],
\end{eqnarray}
where \(W_t=(W^1,\ldots,W^m)\) is a \(m\)-dimensional Brownian motion, \(b:\) \(\DR^d\to \DR^d\) and \(\sigma:\)  \(\DR^d\to \DR^{d}\times \DR^{m} \) are Lipschitz continuous functions.  Although the increments of Wiener process can be simulated exactly, we can consider the following weak Euler scheme
\begin{eqnarray*}
X^{\Delta}_0&=&X_0,
\\
\nonumber
X^{\Delta}_{j\Delta}&=&X^{\Delta}_{(j-1)\Delta}+b\bigl(X^{\Delta}_{(j-1)\Delta}\bigr)\,\Delta+\sum_{k=1}^m\sigma_{k}\bigl(X^{\Delta}_{(j-1)\Delta}\bigr)\,\xi^k_j,
\end{eqnarray*}
where \(j=1,\ldots, n,\) and i.i.d. random variables \((\xi^k_j)\)  satisfy
\begin{eqnarray}
\label{cond_xi_euler}
\bigr|\E[\xi^1_j]\bigl|+\bigl|\E\bigr[(\xi^1_j)^3\bigl]\bigr|+\bigl|\E\bigr[(\xi^1_j)^2\bigl]-\Delta\bigr|\leq c\Delta^2.
\end{eqnarray}
Under some additional assumptions on the coefficient functions \(b\) and \(\sigma,\) and the output function \(f\) spelled out in Talay and Tubaro \cite{talay1990expansion} and Bally and Talay \cite{bally1995euler}, it holds 
\begin{eqnarray*}
\bigl|\E[f(X^\Delta_T)]-\E[f(X_T)]\bigr|\leq c\Delta
\end{eqnarray*}
for some \(c>0.\) The simplest way of constructing a r.v. \(\xi\) with the property \eqref{cond_xi_euler} is to take
\begin{eqnarray}
  \P\left(\xi_{j}^i=\pm \sqrt{\Delta}\right)=\frac{1}{2},\quad i=1,\ldots, m.
  \label{Forms:Binomial:Increments:Base}
\end{eqnarray}
 Observe that distribution of the components of the vector $\bxi_l$ under coupling \eqref{simple_coupling} in the  ML algorithm, is closely related to the Binomial distribution, namely
  \begin{gather}
  \label{incr:dist}
    \dfrac{\xi^i_{l,j}}{2\sqrt{\Delta_L}}+2^{L-l-1}\sim\Bi\left(2^{L-l},\dfrac{1}{2}\right).
  \end{gather}
 Hence the generation of variates $\xi^i_{l,j}$ is straightforward when a generator of binomially
 distributed random variates is available.   For a fixed $L,$ the weak MLMC algorithm implies generation of $\bxi_l$ for $l$ starting from $0$ up to
  $L$. Since all probabilities of distributions $\Bi\left(2^{L-l},\dfrac{1}{2}\right)$ for
  $l\in \{0,1,\ldots,L\}$ are rational numbers, table look-up or alias methods
  (see \cite{devroye1986nonuniform}) can be used to achieve fast single random number generation.
  Since the distributions of $\bxi_l$ do not change between different runs of the MLMC method, all
  the preprocessing required can be done only once and the resulting tables can be stored. In problem-specific
  hardware (FPGA or ASIC) these tables can be kept in permanent shared constant storage, which
  is often cheap, fast and abundant.
  With table lookup methods this would provide $O(L)$ worst-case single random variate generation
  time at the price of storing $O(2^{L+1})$ items of preprocessing data, and with alias methods it is
  possible to attain $O(1)$ worst-case single random variate generation time at the price of storing
  $O(2^{L+1})$ items of preprocessing data.
  Note that the whole procedure of binomial increments generation can be implemented with the use of
  integer numbers only. Therefore, a good MLMC implementation for binomial increments can possibly
  outperform its counterpart for Normal increments.

\subsection{Jump diffusion processes}
\label{jump_sde}
Consider now a \(d\)-dimensional jump diffusion process \((X_t)\) solving the SDE
\begin{eqnarray}
\label{sde:jump}
X_t=X_0+\int_0^t b(X_s)\,ds+\int_0^t \sigma(X_s)\,dW_s+\int_0^t \int_{\mathcal{Z}}\rho(X_{s-},z)\, N(ds,dz),
\end{eqnarray}
where  \(W_t=(W^1_t,\ldots,W^m_t)\)  is a standard \(m\)-dimensional \(\mathcal{F}_t\)-adapted Brownian motion and \(N(ds,dz)\) is a Poisson counting measure on \(\mathbb{R}^+\times \mathcal{Z}\) with a finite intensity measure \(\nu(dz).\)  We assume \(W\) and \(N\) are independent, and that the mappings \(b:\) \(\DR^d\mapsto\DR^d,\)   \(\sigma:\) \(\DR^d\mapsto \DR^{d}\otimes \DR^{m}\) and \(\rho:\)  \(\DR^d\times \mathcal{Z}\mapsto \DR^d\)
are Lipschitz continuous and have at most linear growth on \([0,T]\) so that the solution of \eqref{sde:jump} is well defined.
\par
Let \(\Gamma=\{-1,0,\ldots,m\},\) \(M=\{\gamma=(\gamma_1,\ldots,\gamma_l): \, \gamma_i\in \Gamma,l\geq 0\},\)
and \(\emptyset\) stands for the empty set. For any nonempty \(\gamma=(\gamma_1,\ldots,\gamma_l)\in M,\) denote 
\(\_\gamma=(\gamma_2,\ldots,\gamma_l),\)  \(\gamma\_=(\gamma_1,\ldots,\gamma_{l-1}),\) \(|\gamma|=l,\) 
\(\|\gamma\|=|\gamma|+\) number of zero components of \(\gamma\) and \(\langle\gamma\rangle\)
is the number of negative components of \(\gamma.\)
For any mapping \(\phi\in C^{2}(\DR^d),\) define the operators associated with \eqref{sde:jump}
\begin{eqnarray*}
\mathcal{L}_{-1} [\phi](x;z)&\doteq&\phi(x+\rho(x,z))-\phi(x),
\\
\mathcal{L}_{0} [\phi](x)&\doteq&\sum_{i=1}^d b_i(x)\partial_i \phi(x)+\frac{1}{2} \sum_{i,j=1}^d \varsigma_{ij}(x) \partial_{ij} \phi(x),
\\
\mathcal{L}_{j} [\phi](x)&\doteq&\sum_{i=1}^d \sigma_{ij}(x)\partial_i \phi(x), \quad j=1,\ldots, m,
\end{eqnarray*}
where
\begin{eqnarray*}
\varsigma_{ij}(x)\doteq\sum_{l=1}^m \sigma_{il}(x) \sigma_{jl}(x) \text{ and } \partial_j \doteq \frac{\partial}{\partial x^j}.
\end{eqnarray*}
The composite operator is defined recursively as
\begin{eqnarray*}
\mathcal{L}_{\gamma}[\phi](x;z_1,\ldots,z_{\langle\gamma \rangle})=\mathcal{L}_{\gamma_1}\Bigl(\mathcal{L}_{-\gamma}[\phi](x;z_1,\ldots,z_{\langle\gamma \rangle})\Bigr)
\end{eqnarray*}
if \(\gamma_1\geq 0\) and via
\begin{eqnarray*}
\mathcal{L}_{\gamma}[\phi](x;z_1,\ldots,z_{\langle\gamma \rangle})&=&\mathcal{L}_{\_\gamma}[\phi](x+\rho(x,z_1);z_2,\ldots,z_{\langle\gamma\rangle})
\\
&&-\mathcal{L}_{\_\gamma}[\phi](x;z_2,\ldots,z_{\langle\gamma \rangle})
\end{eqnarray*}
otherwise.
Denote  \(\Delta N_j= N\bigl([(j-1)\Delta,j\Delta],\mathcal{Z}\bigr),\) then the  Euler scheme for \eqref{sde:jump} reads as follows
\begin{eqnarray*}
X^{\Delta}_0&=&X_0,
\\
\nonumber
X^{\Delta}_{j\Delta}&=&X^{\Delta}_{(j-1)\Delta}+a\bigl(X^{\Delta}_{(j-1)\Delta}\bigr)\,\Delta+\sum_{k=1}^m b_{k}\bigl(X^{\Delta}_{(j-1)\Delta}\bigr)\,\Delta W^k_j
\\
\nonumber
&& + \sum_{k=1}^{\Delta N_j}\rho(X^{\Delta}_{(j-1)\Delta},Z_{jk}), \quad j=1,\ldots, n,
\end{eqnarray*}
where \(Z_{jk},\) \(k=1,\ldots, N_j,\) are independent random variables with the law \(\frac{\nu(dz)}{\nu(\mathcal{Z})}.\)
The (essentially) weak Euler scheme can be constructed by replacing the random variables \(\Delta W^k_j\) and \(\Delta N_j\)  by  simple approximations  \(\xi^k_j\) and \(\eta_j,\) respectively which   satisfy
\begin{eqnarray*}
&&\bigr|\E[\xi^k_j]\bigl|+\bigl|\E\bigr[(\xi^k_j)^3\bigl]\bigr|+\bigl|\E\bigr[(\xi^k_j)^2\bigl]-\Delta\bigr|=O(\Delta^2),
\\
&& \bigr|\E[(\eta_j)^l]-\nu(\mathcal{Z})\Delta\bigl|= O(\Delta^2), \quad l=1,2,3.
\end{eqnarray*}
for some \(c>0.\) In particular, one can take 
\begin{eqnarray}
  &&\nonumber\P\left(\xi_{j}^i= \sqrt{\Delta}\right)=\frac{1}{2}, \quad \P\left(\xi_{j}^i= -\sqrt{\Delta}\right)=\frac{1}{2},
  \\
  &&\P\left(\eta_{j}=1\right)=p, \quad  \P\left(\eta_{j}=0\right)=1-p,
   \label{jump:prob} 
\end{eqnarray}
where \(|p-\Delta\,\nu(\mathcal{Z})|=O(\Delta^2).\)
Moreover, the random variables \((Z_{j,k})\) can be replaced by i.i.d. random variables \((\zeta_{j,k})\) satisfying 
\begin{eqnarray}
\label{zeta_moments1}
&&\E\left[L_{\gamma}[\Phi](x,\zeta_{1,l_{1}}\ldots,\zeta_{\langle\gamma\rangle,l_{\langle\gamma\rangle}})\right]=\E\left[L_{\gamma}[\Phi](x,Z_{1,l_{1}}\ldots,Z_{\langle\gamma\rangle,l_{\langle\gamma\rangle}})\right] 
\end{eqnarray}
and
\begin{multline}
\label{zeta_moments2}
\E\left[L_{\gamma}[\Phi](x,\zeta_{1,l_{1}}\ldots,\zeta_{\langle\gamma\rangle,l_{\langle\gamma\rangle}})L^{\top}_{\gamma}[\Phi](x,\zeta_{1,l_{1}}\ldots,\zeta_{\langle\gamma\rangle,l_{\langle\gamma\rangle}})\right]=
\\
=\E\left[L_{\gamma}[\Phi](x,Z_{1,l_{1}}\ldots,Z_{\langle\gamma\rangle,l_{\langle\gamma\rangle}})L^{\top}_{\gamma}[\Phi](x,Z_{1,l_{1}}\ldots,Z_{\langle\gamma\rangle,l_{\langle\gamma\rangle}})\right]
\end{multline}
for \(\Phi(x)\equiv x,\) all \(x\in \DR^d,\) \(|\gamma|\leq 2\) with \(\langle\gamma\rangle>0 \) and \(k=1,2,\) where \(Z\) is distributed according to \(\nu(dz)/\nu(\mathcal{Z}).\) 
\begin{rem}
If \(d=m=1,\)  then coefficients functions of order \(|\gamma|=2\) with \(\langle\gamma\rangle>0\) take the form
\begin{eqnarray*}
L_{(0,-1)}[\Phi](x,\zeta_{1,l_{1}})&=&b(x)\,\partial_x\rho(x,\zeta_{1,l_{1}}),
\\
L_{(1,-1)}[\Phi](x,\zeta_{1,l_{1}})&=&\sigma(x)\,\partial_x\rho(x,\zeta_{1,l_{1}}),
\\
L_{(-1,0)}[\Phi](x,\zeta_{1,l_{1}})&=&b(x+\rho(x,\zeta_{1,l_{1}}))-a(x),
 \\
L_{(-1,1)}[\Phi](x,\zeta_{1,l_{1}})&=&\sigma(x+\rho(x,\zeta_{1,l_{1}}))-b(x),
 \\
L_{(-1,-1)}[\Phi](x,\zeta_{1,l_{1}},\zeta_{1,l_{2}})&=&\rho(x+\rho(x,\zeta_{1,l_{1}}),\zeta_{1,l_{2}})-\rho(x,\zeta_{1,l_{2}}).
\end{eqnarray*}
If moreover  \(\rho(x,u)=x u\),   \(a(x)=a_0+a_1x^{\alpha},\) \(b(x)=b_0+b_1x^{\beta},\)  we get
\begin{eqnarray*}
L_{(0,-1)}[\Phi](x,\zeta_{1,l_{1}})&=&(a_0+a_1x^{\alpha})\zeta_{1,l_{1}},
\\
L_{(1,-1)}[\Phi](x,\zeta_{1,l_{1}})&=& (b_0+b_1x^{\beta}) \zeta_{1,l_{1}},
\\
	L_{(-1,0)}[\Phi](x,\zeta_{1,l_{1}})&=&a_1 x^{\alpha} \left[(1+\zeta_{1,l_{1}})^{\alpha}-1\right],
 \\
L_{(-1,1)}[\Phi](x,\zeta_{1,l_{1}})&=&b_1 x^{\beta} \left[(1+\zeta_{1,l_{1}})^{\beta}-1\right],
 \\
L_{(-1,-1)}[\Phi](x,\zeta_{1,l_{1}},\zeta_{1,l_{2}})&=&x\zeta_{1,l_{1}}\zeta_{1,l_{2}}.
\end{eqnarray*}
Similar results hold for multidimensional case as well.
Hence for a large class of stochastic processes, including affine and polynomial processes, the conditions \eqref{zeta_moments1}
and \eqref{zeta_moments2} can be viewed as generalised moment conditions.
\end{rem}
In the corresponding ML algorithm we can use the approximation \(\eta_{L,j}\) for \(\Delta N_{L,j}\) of the form: 
\begin{eqnarray*}
\P\left(\eta_{L,j}=1\right)=\Delta_L\,\nu(\mathcal{Z}), \quad  \P\left(\eta_{L,j}=0\right)=1-\Delta_L\,\nu(\mathcal{Z}).
\end{eqnarray*}
Then  the random variables \(\eta_{l,j}\) for \(l<L\) have a binomial distribution which can be easily simulated as described in Section~\ref{diffusion}.
\subsection{General L\'evy processes}
Consider a one-dimensional square integrable L\'evy process \((L_t)_{t\geq 0}\) of the form
\begin{eqnarray*}
L_t=b t+\sigma B_t+\int_0^t \int_{\DR} z \tilde N (ds,dz),
\end{eqnarray*}
for some \(\sigma\geq 0,\) where \(\tilde N (ds,dz)\) is a compensated Poisson random measure on \(\DR_+\otimes \DR\) with intensity measure \(ds\,\nu(dz),\) where \(\int |z|^2\,\nu(dz)<\infty.\) In order to apply the Euler approximation scheme to \eqref{sde}, we need to approximate the increments \(\Delta L_j.\) Asmussen and Rosinski \cite{asmussen2001approximations} (see also \cite{jacod2005approximate}) suggested to replace the small jumps in \(L\) by an appropriate Gaussian random variable. So we define 
\[
\zeta_{\Delta,j}^{\delta}\doteq \Delta L^{\delta}_{j}+U^{\delta}_{\Delta,j},
\]
 where \(L^{\delta}\) is the same L\'evy process as \(L\) without its (compensated) jumps smaller than \(\delta\) and \(U^{\delta}_{\Delta,j}\) is Gaussian random variable with the same mean and variance as the neglected jumps. 
 The resulting Euler scheme takes the form
\begin{eqnarray}
\label{euler_scheme_gauss_approx}
X^{\Delta,\delta}_0&=&X_0,
\\
\nonumber
X^{\Delta, \delta}_{j\Delta}&=&X^{\Delta,\delta}_{(j-1)\Delta}+a\bigl(X^{\Delta,\delta}_{(j-1)\Delta}\bigr)\,\zeta_{\Delta,j}^{\delta}, \quad j=1,\ldots, n.
\end{eqnarray}
Let us discuss the first condition in \eqref{ml_ass} (weak convergence). As was shown in \cite{dia2013error} (see also\cite{jacod2005approximate}),
\begin{eqnarray}
\label{bias_gauss_approx}
\Bigl |\E[f(X^{\Delta, \delta}_{T})]-\E[f(X_{T})]\Bigr|\lesssim \Delta\vee\delta^{3-\alpha},
\end{eqnarray}
provided \(f\in C^2(\DR)\) and 
\begin{eqnarray}
\label{nu_0}
\nu\bigl(\{|z|>t\}\bigr)\lesssim t^{-\alpha},\quad t\to +0.
\end{eqnarray}
Note that each r. v. \(\zeta_{\Delta,j}^{\delta}\) can be represented as 
\begin{eqnarray*}
\zeta_{\Delta,j}^{\delta}=\Delta b+\sigma_{\Delta,\delta}\cdot \xi_j+\sum_{i=1}^{N^{\delta}_{\Delta,j}} \bigl(Z_{i,j}^\delta-\E[Z_{i,j}^\delta]\bigr), \quad j=1,\ldots, n,
\end{eqnarray*}
where  \(\sigma^2_{\Delta,\delta}\doteq\Delta \bigl(\sigma^2+\int_{|z|\leq \delta} z^2\, \nu(dz)\bigr),\) \(\xi_j\sim \mathcal{N}(0,1),\) \(N_{\Delta,j}^\delta\sim \operatorname{Poiss}\Bigl(\Delta \, \nu\bigl(\{|z|>\delta\}\bigr)\Bigr)\) and \(Z_{1,j}^\delta,Z_{2,j}^\delta,\ldots \) are i.i.d. random variables with the distribution 
\[
1_{|z|>\delta}\,\nu(dz)/\nu\bigl(\{|z|>\delta\}\bigr).
\]
Hence the cost of generating one trajectory by means of \eqref{euler_scheme_gauss_approx} is of order \(\Delta^{-1}+\nu\bigl(\{|z|>\delta\}\bigr).\) Let us now fix two natural numbers \(n_f,\) \(n_c=2\cdot n_f,\) two positive real numbers \(\delta_f,\) \(\delta_c\) and describe the coupling between 
\(\zeta_{\Delta_c}^{\delta_c}\) and \(\zeta_{\Delta_f}^{\delta_f}.\) Set
\begin{multline}
\zeta_{\Delta_c,j}^{\delta_c}=2\Delta_f b+\sigma_{\Delta_f,\delta_f}\cdot (\xi_{2j}+\xi_{2j-1})
\\
+\sum_{i=1}^{N^{\delta_f}_{\Delta_f,2j}} \Bigl(Z_{i,2j}^{\delta_f} \, 1_{|Z_{i,2j}|>\delta_c}-\E\Bigl[Z_{i,2j}^{\delta_f} \, 1_{|Z_{i,2j}|>\delta_c}\Bigr]\Bigr)
\\
\label{coupling_1}
+\sum_{i=1}^{N^{\delta_f}_{\Delta_f,2j-1}} \Bigl(Z_{i,2j-1}^{\delta_f} \, 1_{|Z_{i,2j-1}|>\delta_c}-\E\Bigl[Z_{i,2j-1}^{\delta_f} \, 1_{|Z_{i,2j-1}|>\delta_c}\Bigr]\Bigr),
\end{multline}
then 
\begin{eqnarray*}
\mathcal{R}&\doteq&\zeta_{\Delta_c,j}^{\delta_c}-\zeta_{\Delta_f,2j}^{\delta_f}-\zeta_{\Delta_f,2j-1}^{\delta_f}
\\
&=& \sum_{i=1}^{N^{\delta_f}_{\Delta_f,2j}} \Bigl(Z_{i,2j}^{\delta_f} \, 1_{|Z_{i,2j}|\leq\delta_c}-\E\Bigl[Z_{i,2j}^{\delta_f} \, 1_{|Z_{i,2j}|\leq\delta_c}\Bigr]\Bigr)
\\
&& +\sum_{i=1}^{N^{\delta_f}_{\Delta_f,2j-1}} \Bigl(Z_{i,2j-1}^{\delta_f} \, 1_{|Z_{i,2j-1}|\leq\delta_c}-\E\Bigl[Z_{i,2j-1}^{\delta_f} \, 1_{|Z_{i,2j-1}|\leq\delta_c}\Bigr]\Bigr).
\end{eqnarray*}
As a result, \(\E[\mathcal{R}]=0\) and
\begin{eqnarray*}
\E\bigl[|\mathcal{R}|^2\bigr]\leq 2\Delta_f \int_{|z|\leq \delta_c}|z|^2\,\nu(dz).
\end{eqnarray*}
Hence the assumptions of Corollary~\ref{cor_main} are fulfilled, provided 
\begin{eqnarray*}
\int_{|z|\leq \delta_c}|z|^2\,\nu(dz)\leq c\Delta_f
\end{eqnarray*}
for some \(c>0.\) Under \eqref{nu_0}, this is equivalent to the relation \(\delta_c\lesssim \Delta^{1/(2-\alpha)}_f.\)
Using the estimate \eqref{bias_gauss_approx}, we derive the complexity of the resulting coupled multilevel scheme.
\begin{prop}
\label{compl_gen_levy}
If \(\alpha\leq 3-\sqrt{3}\) in \eqref{nu_0}, then the complexity of the coupled multilevel algorithm presented in Section~\ref{mlmc_alg} with the coupling  \eqref{coupling_1} is of order 
\[
\begin{cases}
 \varepsilon^{-2}\cdot\left(\log{\varepsilon}\right)^2, & \alpha\le1,\\
  \varepsilon^{-\frac{2}{2-\alpha}}, & 1<\alpha\le 3-\sqrt{3},\\
 \end{cases}
 \]
provided \(\delta_l=\Delta^{1/(2-\alpha)}_l.\) For \(\alpha>3-\sqrt{3}\) we can use the simplest coupling \(\zeta_{\Delta_c,j}^{\delta_c}=\zeta_{\Delta_f,2j}^{\delta_f}+\zeta_{\Delta_f,2j-1}^{\delta_f}\) and constant \(\delta_l=\varepsilon^{1/(3-\alpha)}\) to get  upper estimate \(\varepsilon^{-(6 - \alpha)/(3-\alpha)} \) for the complexity of the corresponding coupled multilevel algorithm.
\end{prop} 
\paragraph{Discussion}
Observe that the complexity of the standard MC algorithm  for estimating \(\E[f(X_T)]\) is bounded above via
\[
\begin{cases}
 \varepsilon^{-3}, & \alpha\le 3/2,\\
  \varepsilon^{-(6 - \alpha)/(3-\alpha)}, & 3/2<\alpha\le 2.
 \end{cases}
 \]
So the coupled MLMC approach is superior to the standard MC algorithm as long as \(\alpha\leq 3-\sqrt{3}.\)  A similar behaviour can be observed  
in Dereich \cite{dereich2011multilevel} (at least  for \(\alpha\leq 1\)).
We can further replace the restricted L\'evy jump sizes \((Z_{i,j}^{\delta_f})\) by some simple random variables using the approach presented in Section~\ref{jump_sde}. Note that in the latter case the above complexity bounds continue to hold. 
\section{Numerical experiments}
In this section we present numerical examples corresponding to process classes discussed in Section~\ref{examples}.

The MLMC algorithm is implemented according to the \cite{giles2008multilevel}, with some changes, due to the specific structure of the simulated process. Recall, that the MLMC estimator has the form:
\begin{align*}
Y_{L,\mathbf{N}}&=\frac{1}{N_{0}}\sum_{n=1}^{N_{0}}\left[f(X_{T}^{0}(\boldsymbol{\zeta}^{(n)})\right]+
\sum_{l=1}^{L}\frac{1}{N_{l}}\sum_{n=1}^{N_{l}}\left[f\bigl(X_{T}^{l}(\boldsymbol{\zeta}_l^{(n)})\bigr)-f\bigl(X_{T}^{l-1}\bigl(\overline{\boldsymbol{\zeta}}_l^{(n)}\bigr)\bigr)\right],\\
&=\hat{Y}_0 + \sum_{l=0}^{L} \hat{Y}_l
\end{align*}
But the general scheme is the same for all considered problems and can be summarized in the following algorithm:
\begin{enumerate}
 \item[Input:] Requested accuracy $\epsilon$ and set the final level $\hat{L}$.
 \item Set $L:=2$
 \item Compute $N_l:=100$ samples on levels $l=0,1,2$
\item Estimate $\Var(\hat{Y}_l)$ and update $N_l$ for each level $l=0,\dots,L$:
$$N_l := \max\left\{N_l, \left\lceil 2\cdot\epsilon^{-2} \cdot \sqrt{\Var(\hat{Y}_l)  2^{-l}} \cdot \sum_{k=0}^L \sqrt{\Var(\hat{Y}_k)\cdot 2^{k}}\right\rceil\right\}$$
If the update $N_l$ is increased less than 1\% on the levels, then go to step 5.  
\item Compute the additional number of samples and Go to step 3.
\item If $L<2$ or $\max\{|\hat{Y}_{L-1}|/2, |\hat{Y}_L|\} \geq \epsilon/\sqrt{2}$:\\
$\hphantom{mm}L:=L+1$ , $\Var(\hat{Y}_L) =\Var(\hat{Y}_{L-1})/ 2$\\
Else: Return $\sum_{l=0}^L \hat{Y}_l$.
\item If $L>\hat{L}$, then \\
Display error: The final level $\hat{L}$ is insufficient for the convergence.\\
Return $\sum_{l=0}^L\hat{Y}_l$.
\item Goto step 3.
\end{enumerate}
In all of our numerical experiments we have chosen $\hat{L}$ to be sufficiently large, so that $\hat{L}\ge L$ was always satisfied.
\subsection{Diffusion process}
\label{sec:num:2}
\subsubsection{European max-call option}
Consider a three dimensional process $X_t=(X_t^1,X^2_t,X_t^3),$ \(t\in [0,T],\) with independent
components where each process \(X^i_t\) solves one-dimensional SDE of the form \eqref{sde:diffusion} with  
$b(x) = r\cdot x$ and $\sigma(x) = \sigma\cdot x$ for some $r,\,\sigma\in\mathbb{R}$. We are
interested in computing the expectation of
\begin{eqnarray*}
  f(X_T) = e^{-r\cdot T}\max\left(\max(X^1_{T},\,X^2_{T},\,X^3_{T})-K, 0\right).
\end{eqnarray*}
We chose the following parameters:
\begin{eqnarray*}
  r = 0.05, \quad \sigma = 0.2, \quad T = 1, \quad K = 1, \quad X_0^i = 1,\quad i=1,2,3,
\end{eqnarray*}
and $\hat{L}=9$.
In fact in this case the exact solution is available and for above parameter values, we have $\E [f(X_T)]\approx 0.2276799594$.
The variance decay is presented on Figure~\ref{Figs:Norm:Bin:Comp2}. In particular, the line  $\alpha_1-\alpha_2\cdot l$ with $\alpha_2 = 0.9753$
fits the estimated log-variances best and this is in agreement with Corollary \ref{cor_main}.
The corresponding RMSE is presented in Figure~\ref{Figs:RMSE2}.

\begin{figure}[ht]
  \centerline{\hfill\includegraphics[width=0.75\textwidth]{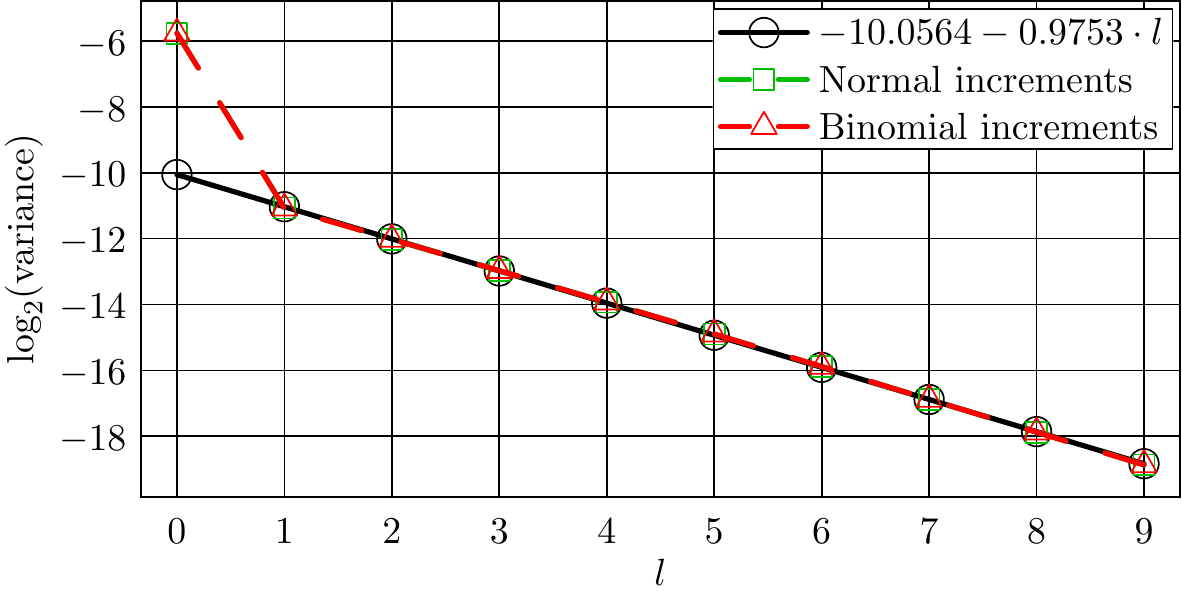}\hfill}
  \caption{Three dimensional European max-call option: level variances for schemes with binomial and normal increments.}
  \label{Figs:Norm:Bin:Comp2}
\end{figure}

\begin{figure}[ht]
  \centerline{\hfill\includegraphics[width=0.75\textwidth]{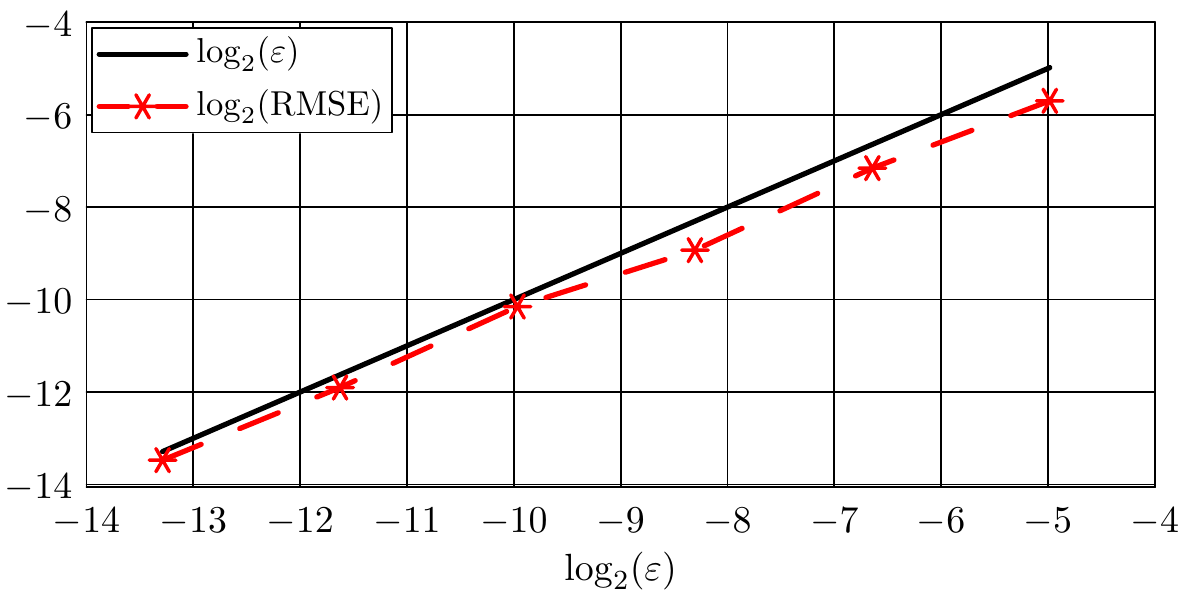}\hfill}
  \caption{Three dimensional European max-call option: estimated RMSE against the required precision \(\varepsilon\) for different values of \(\varepsilon\).}
  \label{Figs:RMSE2}
\end{figure}
\subsubsection{Geometric Asian option} Consider a one dimensional process $X_t,$ \(t\in [0,T],\) where each coordinate process \(X^i_t\) solves one-dimensional SDE of the form \eqref{sde:diffusion} with 
$b(x) = r\cdot x$ and $\sigma(x) = \sigma\cdot x$ for some $r,\,\sigma\in\mathbb{R}$. We are
interested in computing the expectation of the functional
\begin{eqnarray*}
  f(X_\cdot) = e^{-r\cdot T}\max\left(\exp\left(\frac1T\int\limits_{0}^T\log(X_t)dt\right)-K, 0\right).
\end{eqnarray*}
The parameter values are
\begin{eqnarray*}
r = 0.05, \ \sigma = 0.2,\ \  T = 1, \  K = 1.
\end{eqnarray*}
In this case the exact value of the expectation is given by 
\[\E[f(X_\cdot)]\approx 0.05546818634.\]
The variance decay is presented on Figure~\ref{Figs:Norm:Bin:Asian}. Due the fact, that at first levels the variance decays faster than predicted, we have fitted the variance decay only on the last 6 levels with the line  $\alpha_1-\alpha_2\cdot l$ and got $\alpha_2 = 1.0059$.
The corresponding RMSE is presented in Figure~\ref{Figs:RMSE:Asian}.

\begin{figure}[ht]
  \centerline{\hfill\includegraphics[width=0.75\textwidth]{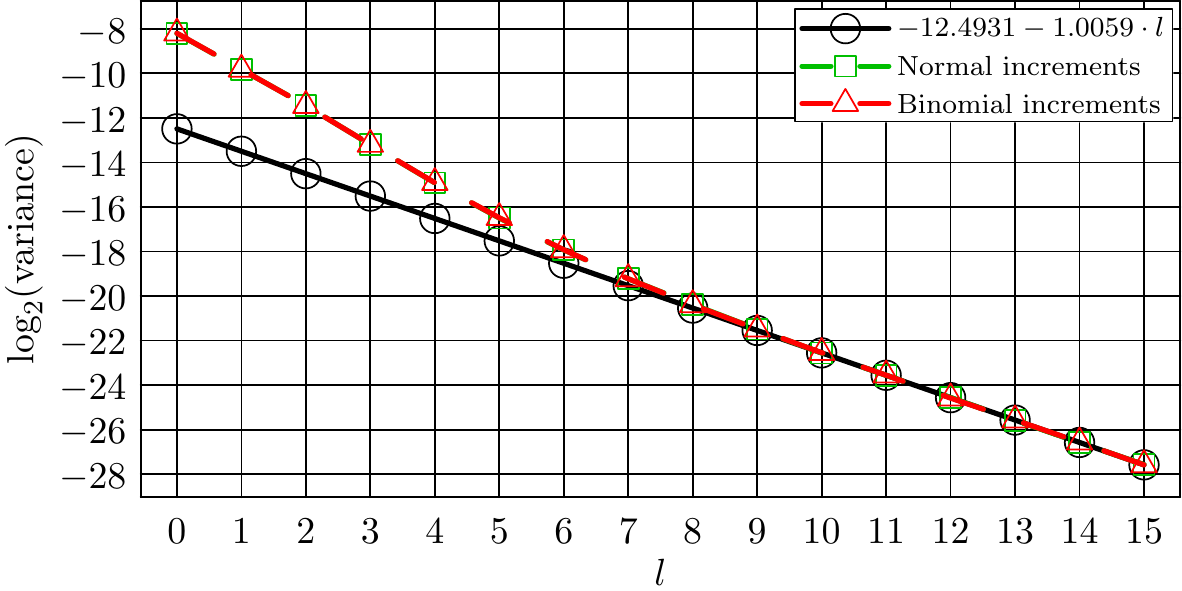}\hfill}
  \caption{Geometric Asian option: level log-variances for binomial and normal increments.}
  \label{Figs:Norm:Bin:Asian}
\end{figure}

\begin{figure}[ht]
  \centerline{\hfill\includegraphics[width=0.75\textwidth]{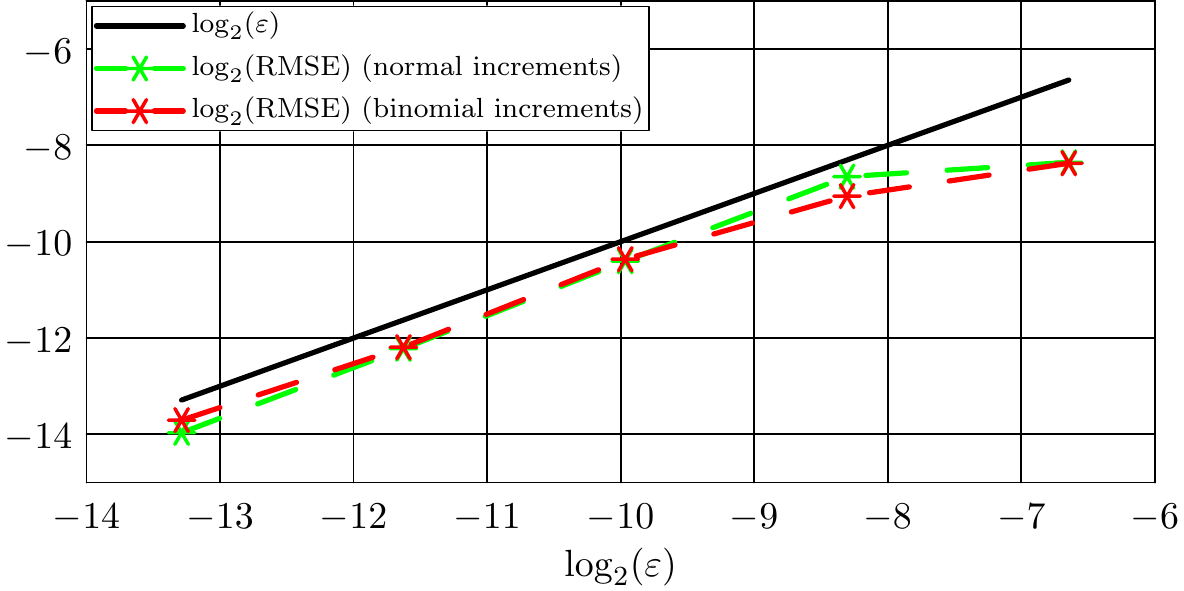}\hfill}
  \caption{Geometric Asian option: RMSE for binomial and normal increments.}
  \label{Figs:RMSE:Asian}
\end{figure}

%%%%%%%%%%%%            SDE with jumps
\subsection{Jump diffusions}
\label{sec:num:3}
Consider a jump SDE
\begin{equation*}
dX_{t}=\left(r-\lambda\cdot \left(e^{m+0.5\cdot\theta^2}-1\right)\right)\cdot X_{t}\cdot dt+\sigma\cdot X_{t}\cdot W_{t} + X_{t}\cdot dJ(t),
\end{equation*}
where 
$$J(t)=\sum\limits_{j=1}^{N(t)}(Y_j-1),\quad \log(Y_j)\sim \mathcal{N}(m, \theta^2)$$
and $N(t)$ is a Poisson process with rate $\lambda$.
We are
interested in computing the expectation of
\begin{eqnarray*}
  f(X_T) = e^{-r\cdot T}\max\left(X(T)-K, 0\right).
\end{eqnarray*}
The parameters' values are
\begin{eqnarray*}
r = 0.05, \ \sigma = 0.2,\  \lambda = 0.5,\ m = 0.05,\  \theta = 0.25,\  T = 1, \  K = 1.
\end{eqnarray*}
It follows from \cite{glasserman2004monte} (Section 3.5) that, for above  parameter values $\E[f(X_T)] \approx0.153065585$. We have performed two types of simulations with the fixed top level $\hat{L}=8$:
\begin{itemize}
\item $Y$ was sampled from the lognormal distribution, while the increments of the Brownian motion were modelled as normal random variables
\item $Y$ was sampled as a discrete random variable $\hat{Y}$ according to the Remark \ref{zeta_moments2}, with moments matching first $6$ moments of the lognormal distribution, while the increments of the Brownian motion were modelled as discrete random variable defined by  \eqref{incr:dist}.
\end{itemize}
In both of those cases, the number of jumps $\eta_{l,j}$ at the level $l$ and step $j$
is generated  via
  \begin{gather}
  \label{incr:jump}
    \eta_{l,j}\sim\Bi\left(2^{\hat{L}-l},2^{-\hat{L}}\cdot \lambda\right).
  \end{gather}
One can see, that \eqref{incr:jump} can be implemented  in the same spirit as \eqref{incr:dist}.
On the finest level $L,$ we allow only for two jumps $0$ or $1.$  
Let us denote by $\mu_i$ the $i$th moment of the lognormal distribution with parameters $m$ and $\theta$. The random variable $\hat{Y}$ takes $4$ values with probabilities $p_1,\ldots,p_4$. The values and probabilities are obtained by solving the optimization problem:
\begin{align*}
\text{Minimize }&\\
& \left(\sum_{k=1}^4 p_k\cdot x_k^7 - \mu_7\right)^2\\
\text{Subject to }&\\
&\sum_{k=1}^4 p_k\cdot x_k^i = \mu_i,\qquad i=1,\ldots,6
\end{align*}
The solution is
\begin{align*}
p_1 = 0.608176614910593,  &\quad  x_1 = 1.081500568717563 \\ 
p_2 = 0.003503326771883,  &\quad  x_2 = 2.376117006693613 \\
p_3 = 0.226782660300013,  &\quad  x_3 = 0.719559222085786 \\
p_4 = 0.161537398017512,  &\quad  x_4 = 1.581001071314797
\end{align*}
The variance decay is shown in Figure~\ref{Figs:Merton:Max:Variance} for both types of simulations. We estimated RMSE of the
ML estimate based on the weak Euler scheme based on \(50\) independent runs, see
Figure~\ref{Figs:Merton:Max:RMSE}.

\begin{figure}
  \centerline{\hfill\includegraphics[width=0.75\textwidth]{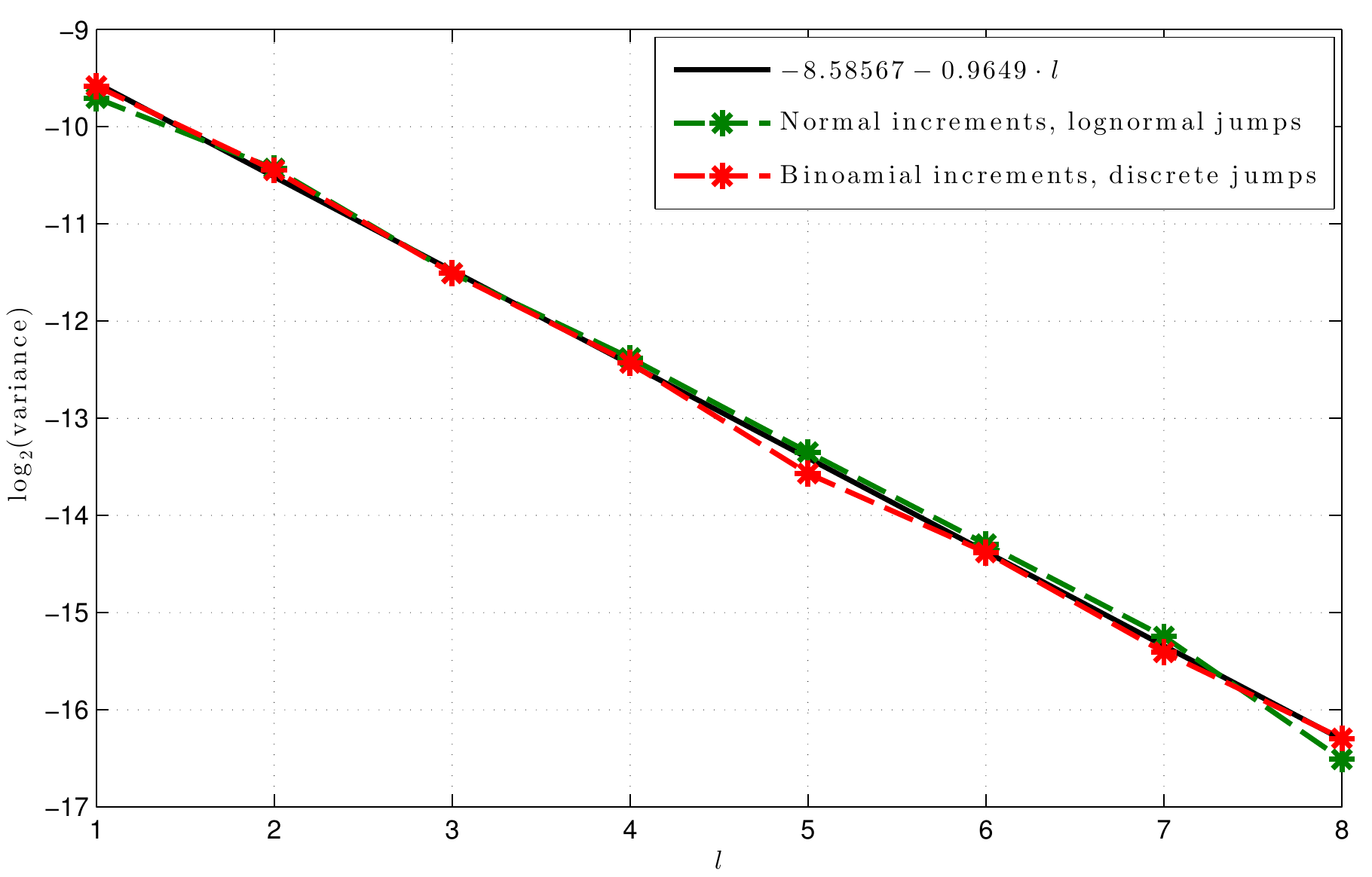}\hfill}%Variance_Milstein_3d_linear.pdf}\hfill}
  \caption{European option: level log-variances for binomial increments and discrete jumps, normal increments and lognormal jumps.} %Comparison of level variances% for schemes with Normal and $\bxi$-noise}
  \label{Figs:Merton:Max:Variance}
\end{figure}
\begin{figure}[ht]
  \centerline{\hfill\includegraphics[width=0.75\textwidth]{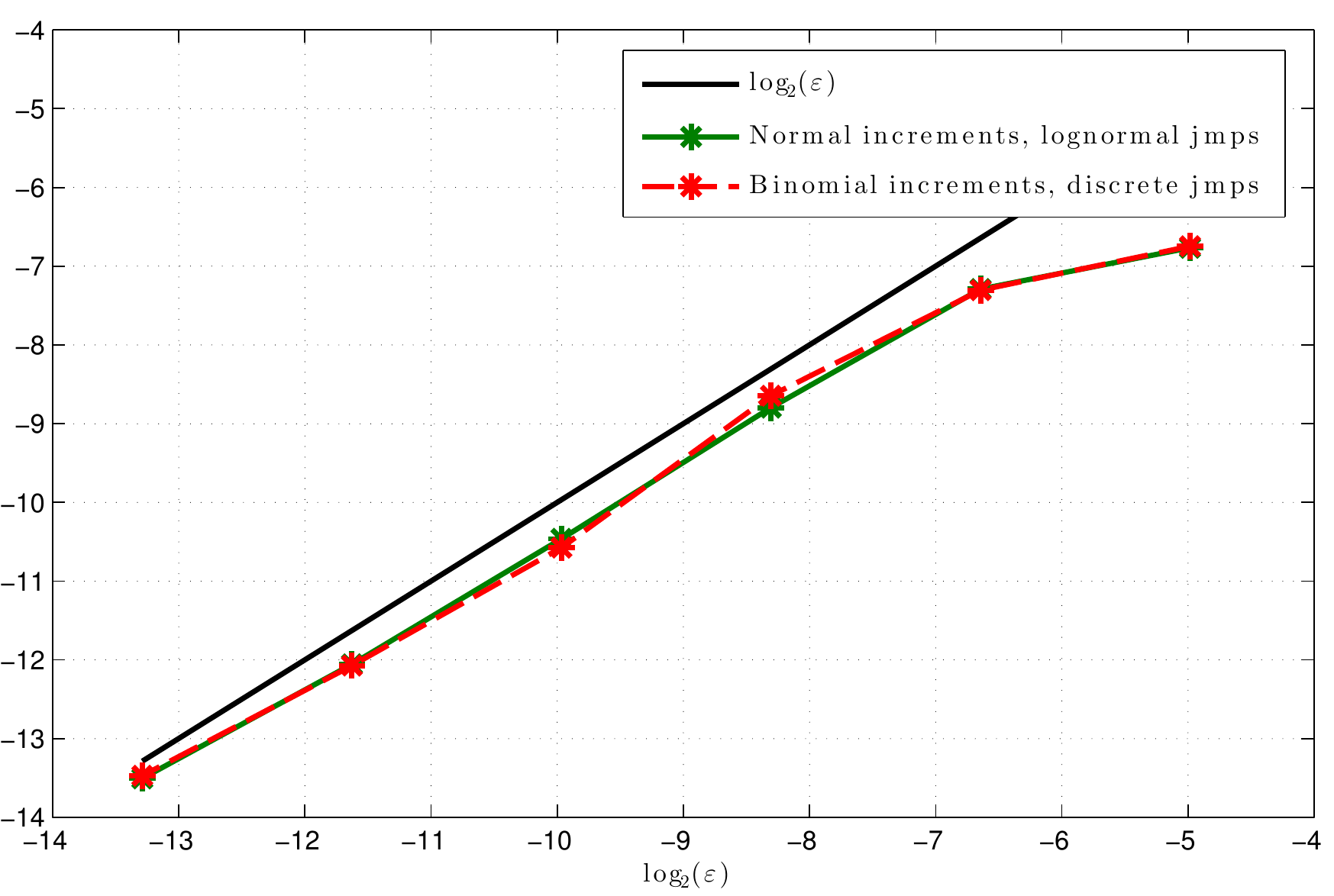}\hfill}%Variance_Milstein_3d_linear.pdf}\hfill}
  \caption{European option: RMSE for binomial increments and discrete jumps, normal increments and lognormal jumps.} %Comparison of level variances% for schemes with Normal and $\bxi$-noise}
  \label{Figs:Merton:Max:RMSE}
\end{figure}

%%%%%%%%%%%%                 Milstein scheme

\paragraph{Acknowledgments}
Authors are thankful to Prof.~Mike Giles, Dr.~Lukasz Szpruch and Dr.~Sonja Cox for their helpful comments and remarks. Authors are very grateful to Vladimir Shiryaev for his assistance with the numerical experiments.
\section{Proofs}
\begin{lem}
\label{bound}
Suppose that the coefficient function \(a\) in \eqref{sde} is uniformly Lipschitz and has at most linear growth, i.e.,
\begin{eqnarray}
\label{lip}
\|a(x)-a(x')\|\leq L_a\,\|x-x'\|,\quad 
\|a(x)\|^2\leq B^2_a \, (1+\|x\|^2)
\end{eqnarray}
for any \(x,x'\in \DR^d\)  and some positive constants \(L_a\) and \(B_a.\) 
Moreover, assume that \(\E[\|X_0\|^2]<\infty,\) then the following estimates hold
$$\E\left[\bigl\Vert X_{n\Delta_f}^f \bigr\Vert^2\right]\leq 3 B_{a}^{2}\cdot (n\cdot m_{f,2}+n^2\cdot m^2_{f,1})\cdot\exp\left(3B_{a}^{2}\cdot (n\cdot m_{f,2}+n^2\cdot m^2_{f,1})\right),$$
$$\E\left[\bigl\Vert X_{n\Delta_c}^c \bigr\Vert^2\right]\leq 3 B_{a}^{2}\cdot (n\cdot m_{f,2}+n^2\cdot m^2_{f,1})\cdot\exp\left(3B_{a}^{2}\cdot (n\cdot m_{f,2}+n^2\cdot m^2_{f,1})\right),$$
for \(n=1,\ldots, n_c.\)
\end{lem}
\begin{proof} 
Since
\begin{align*}
X_{n\Delta_f}^f&= X_0+\sum\limits_{i=1}^n \left(X_{i\Delta_f}^f - X_{(i-1)\Delta_f}^f\right),
\end{align*}
we have, due to independence of the increments 
\begin{align*}
\E\left[\bigl\Vert X_{n\Delta_f}^f \bigr\Vert^2\right] &= \E\left[\Bigl\Vert X_0+\sum\limits_{i=1}^n a\bigl(X_{(i-1)\Delta_f}^f\bigr)\cdot (\zeta_i^f - \E [\zeta_{i}^{f}]\bigr)+\sum\limits_{i=1}^n a\bigl(X_{(i-1)\Delta_f}^f\bigr)\cdot \E[\zeta_{i}^{f}]\Bigr\Vert^2 \right]\\
&\le 3\,\E[\|X_0\|^2]+ 3\,\sum\limits_{i=1}^n\mathrm{E}\left[\bigl\Vert  a\bigl(X_{(i-1)\Delta_f}^f\bigr)\bigr\Vert  ^{2}\right]\, m_{f,2}\\
&\quad +3n\,\sum\limits_{i=1}^n\mathrm{E}\Bigl[\bigl\Vert  a\bigl(X_{(i-1)\Delta_f}^f\bigr)\bigr\Vert  ^{2}\Bigr]\,m^2_{f,1} \\
&\le 3\E[\|X_0\|^2]+3B_{a}^{2}\cdot m_{f,2} \cdot\Bigl (n  + \sum\limits_{i=1}^n \mathrm{E}\Bigl[\bigl\Vert  X^f_{i-1}\bigr\Vert  ^{2}\Bigr]\Bigr)\\
&\quad + 3B_{a}^{2}\cdot n\cdot m^2_{f,1} \cdot\Bigl (n  + \sum\limits_{i=1}^n \mathrm{E}\Bigl[\bigl\Vert  X^f_{i-1}\bigr\Vert  ^{2}\Bigr]\Bigr)
\end{align*}
Using the discrete version of the Gronwall inequality (see Appendix), we get 
\[
\E\left[\bigl\Vert X_{n\Delta_f}^f\bigr\Vert^2\right]\leq 3 B_{a}^{2}\cdot (n\cdot m_{f,2}+n^2\cdot m^2_{f,1})\cdot\exp\left(3B_{a}^{2}\cdot (n\cdot m_{f,2}+n^2\cdot m^2_{f,1})\right).
\]
The second inequality of the lemma is proved in the same way.
\end{proof}
\subsection{Proof of Proposition ~\ref{main_prop}}

%First note that 
%\begin{eqnarray*}
% \mathrm{E}\left[\bigl\Vert  X_{\Delta_{f}}^{f}-X_{0}\bigr\Vert  ^{2}\right]	&\leq &	\mathrm{E}\Bigl[\bigl\Vert  a(X_{0})\bigr\Vert  ^{2}\Bigr]\mathrm{E}\left[\bigl\Vert  \zeta_1^f\bigr\Vert  ^{2}\right]
% \\
% &\leq &	B_{a}^{2}\left(1+\mathrm{E}\Bigl[\bigl\Vert  X_{0}\bigr\Vert  ^{2}\Bigr]\right)m_{f,2}
% \end{eqnarray*}
% with $m_{f,2}\doteq\mathrm{E}\left[\bigl\Vert  \zeta_{j}^{f}\bigr\Vert  ^{2}\right].$
%Hence \(\mathrm{E}\Bigl[\bigl\Vert X_{\Delta_{f}}^{f}\bigr\Vert ^{2}\Bigr]<\infty,\) provided \(\mathrm{E}\Bigl[\bigl\Vert X_{0}\bigr\Vert ^{2}\Bigr]<\infty.\) Analogously \(\mathrm{E}\Bigl[\bigl\Vert {X}_{\Delta_{c}}^{c}\bigr\Vert ^{2}\Bigr]<\infty.\)
Due to the Lemma \ref{bound} we have
\begin{eqnarray}
\label{moments}
\mathrm{E}\Bigl[\bigl\Vert X_{n\Delta_{f}}^{f}\bigr\Vert ^{2}\Bigr]<A_1, \quad \mathrm{E}\Bigl[\bigl\Vert {X}_{n\Delta_{c}}^{c}\bigr\Vert ^{2}\Bigr]<A_2
\end{eqnarray}
for $n=1,\ldots,n_{c},$ and constants \(A_1,\) \(A_2\) not depending on \(n.\)  
We have
\begin{eqnarray*}
X_{r\Delta_{c}}^{f}-{X}_{r\Delta_{c}}^{c} & = & X_{(r-1)\Delta_{c}}^{f}-{X}_{(r-1)\Delta_{c}}^{c}+\left[a(X_{(2r-1)\Delta_{f}}^{f})-a({X}_{(r-1)\Delta_{c}}^{c})\right]\zeta_{2r}^{f}+\\
 &  & +\left[a(X_{(r-1)\Delta_{c}}^{f})-a({X}_{(r-1)\Delta_{c}}^{c})\right]\zeta_{2r-1}^{f}
 -a(X_{(r-1)\Delta_{c}}^{c})\left[{\zeta}_{r}^{c}-\zeta_{2r}^{f}-\zeta_{2r-1}^{f}\right]
\end{eqnarray*}
Denote  $D_{r}\doteq X_{r\Delta_{c}}^{f}-{X}_{r\Delta_{c}}^{c},$
then we have the representation 
\[
D_{r}=D_{r-1}+\delta_{r}+\varepsilon_{r}
\]
with 
\begin{eqnarray*}
\varepsilon_{r}&=&\left[a(X_{(2r-1)\Delta_{f}}^{f})-a({X}_{(r-1)\Delta_{c}}^{c})\right]\left(\zeta_{2r}^{f}-\E\bigl [\zeta_{2r}^{f}\bigr]\right)
\\
&&+\left[a(X_{(r-1)\Delta_{c}}^{f})-a({X}_{(r-1)\Delta_{c}}^{c})\right]\left(\zeta_{2r-1}^{f}-\E\bigl [\zeta_{2r-1}^{f}\bigr]\right)
\\
&& -a({X}_{(r-1)\Delta_{c}}^{c})\left[{\zeta}_{r}^{c}-\zeta_{2r}^{f}-\zeta_{2r-1}^{f}\right]
\end{eqnarray*}
and 
\begin{eqnarray*}
\delta_{r}&=&\left[a(X_{(2r-1)\Delta_{f}}^{f})-a({X}_{(r-1)\Delta_{c}}^{c})\right]\E\bigl [\zeta_{2r}^{f}\bigr]+\left[a(X_{(r-1)\Delta_{c}}^{f})-a({X}_{(r-1)\Delta_{c}}^{c})\right]\E\bigl [\zeta_{2r-1}^{f}\bigr].
\end{eqnarray*}
The Lipschitz continuity of the function $a$ implies
\begin{eqnarray*}
\mathrm{E}\left[\bigl\Vert a(X_{(2r-1)\Delta_{f}}^{f})-a({X}_{(r-1)\Delta_{c}}^{c})\bigr\Vert ^{2}\right] & \leq & L_{a}^{2}\,\mathrm{E}\left[\bigl\Vert X_{(2r-1)\Delta_{f}}^{f}-{X}_{(r-1)\Delta_{c}}^{c}\bigr\Vert ^{2}\right]\\
 & \leq & 2\,L_{a}^{2}\,\mathrm{E}\left[\bigl\Vert D_{r-1}\bigr\Vert ^{2}\right]\\
 &  & +2\,L_{a}^{2}\,B_{a}^{2}\,\left(1+\mathrm{E}\bigl\Vert X_{(r-1)\Delta_{c}}^{f}\bigr\Vert ^{2}\right)\mathrm{E}\left[\bigl\Vert \zeta_{2r-1}^{f}\bigr\Vert ^{2}\right]
\end{eqnarray*}
and
\[
\mathrm{E}\left[\bigl\Vert  a(X_{(r-1)\Delta_{c}}^{f})-a({X}_{(r-1)\Delta_{c}}^{c})\bigr\Vert  ^{2}\right]\leq L_{a}^{2}\,\mathrm{E}\left[\bigl\Vert  D_{r-1}\bigr\Vert  ^{2}\right].
\]
As a result 
\begin{eqnarray*}
\mathrm{E}\left[\bigl\Vert  \varepsilon_{r}\bigr\Vert  ^{2}\right] & \leq & 3\mathrm{E}\left[\bigl\Vert  a(X_{(2r-1)\Delta_{f}}^{f})-a({X}_{(r-1)\Delta_{c}}^{c})\bigr\Vert  ^{2}\right]\mathrm{E}\left[\bigl\Vert  \xi_{2r}^{f}\bigr\Vert  ^{2}\right]+\\
 &  & +3\mathrm{E}\left[\bigl\Vert  a(X_{(r-1)\Delta_{c}}^{f})-a({X}_{(r-1)\Delta_{c}}^{c})\bigr\Vert  ^{2}\right]\mathrm{E}\left[\bigl\Vert  \xi_{2r-1}^{f}\bigr\Vert  ^{2}\right]\\
 && +3\mathrm{E}\left[\bigl\Vert  a({X}_{(r-1)\Delta_{c}}^{c})\bigr\Vert  ^{2}\right]\mathrm{E}\left[\bigl\Vert  {\zeta}_{r}^{c}-\zeta_{2r}^{f}-\zeta_{2r-1}^{f}\bigr\Vert  ^{2}\right]
 \\
 & \leq & c_{1}\left(1+\mathrm{E}\bigl\Vert  X_{(r-1)\Delta_{c}}^{f}\bigr\Vert  ^{2}\right)\mathrm{E}\left[\bigl\Vert  \xi_{2r-1}^{f}\bigr\Vert  ^{2}\right]\mathrm{E}\left[\bigl\Vert  \xi_{2r}^{f}\bigr\Vert  ^{2}\right]+\\
 &  & +c_{2}\mathrm{E}\left[\bigl\Vert  D_{r-1}\bigr\Vert  ^{2}\right]\left(\mathrm{E}\left[\bigl\Vert  \xi_{2r}^{f}\bigr\Vert  ^{2}\right]+\mathrm{E}\left[\bigl\Vert  \xi_{2r-1}^{f}\bigr\Vert  ^{2}\right]\right)\\
& & + c_3 \left(1+\mathrm{E}\left[\bigl\Vert  X_{(r-1)\Delta_{c}}^{c}\bigr\Vert  ^{2}\right]\right) \mathcal{R}
\\
 & \leq & c_{4}\left[m_{f,2}\mathrm{E}\left[\bigl\Vert  D_{r-1}\bigr\Vert  ^{2}\right]+m_{f,2}^{2}+\mathcal{R}\right]
\end{eqnarray*}
for some constants $c_{1},c_{2},$  $c_{3},$ \(c_4\) and \(\mathcal{R}\doteq\mathrm{E}\left[\bigl\Vert {\zeta}_{r}^{c}-\zeta_{2r}^{f}+\zeta_{2r-1}^{f}\bigr\Vert  ^{2}\right].\) 
 Analogously 
\[
\mathrm{E}\left[\bigl\Vert  \delta_{r}\bigr\Vert  ^{2}\right]\leq c_{5}\left[m_{f,1}^{2}\mathrm{E}\left[\bigl\Vert  D_{r-1}\bigr\Vert  ^{2}\right]+m_{f,1}^{2}m_{f,2}\right]
\]
for  some $c_{5}>0.$ Define 
\[
M_{r}=\sum_{j=1}^{r}\varepsilon_{j}
\]
and note that $M_{r}$ is martingale with respect to the filtration 
\[
\mathcal{F}_{r}\doteq\sigma\left(X_{(2j-1)\Delta_{f}}^{f},{X}_{(j-1)\Delta_{c}}^{c},X_{(j-1)\Delta_{c}}^{f},\, j\leq r\right),\quad r=1,\ldots,n_{c}+1.
\]
 Hence the Doob inequality implies for any $n\leq n_{c}:$
\begin{eqnarray*}
\mathrm{E}\left[\bigl\Vert  \sup_{r=1,\ldots,n}M_{r}\bigr\Vert  ^{2}\right]&\leq&\mathrm{E}\left[\bigl\Vert  M_{n}\bigr\Vert  ^{2}\right]
\\
&\leq & c_{6}m_{f,2}\sum_{j=1}^{n}\mathrm{E}\left[\bigl\Vert  D_{j-1}\bigr\Vert  ^{2}\right]+c_{6}nm_{f,2}^{2}
+c_6 n \, \mathcal{R}.
\end{eqnarray*}
So we have for $\overline{D}_{n}\doteq\max_{j=1,\ldots,n}D_{j}$ 
\begin{eqnarray*}
\mathrm{E}\left[\bigl\Vert  \overline{D}_{n}\bigr\Vert  ^{2}\right] & \leq & 2\cdot n\sum_{j=1}^{n}\mathrm{E}\left[\bigl\Vert  \delta_{j}\bigr\Vert  ^{2}\right]+2\cdot \mathrm{E}\left[\bigl\Vert  M_{n}\bigr\Vert  ^{2}\right]\\
 & \leq & c_{7}\left(m_{f,2}+nm_{f,1}^{2}\right)\sum_{j=1}^{n}\mathrm{E}\left[\bigl\Vert  \overline{D}_{j-1}\bigr\Vert  ^{2}\right]+c_{7}n\left(m_{f,2}^{2}+n m_{f,1}^{2}m_{f,2}+\mathcal{R}\right)
\end{eqnarray*}
Finally a discrete version of Gronwall lemma (see Appendix) implies 
\[
\mathrm{E}\left[\bigl\Vert  \overline{D}_{n}\bigr\Vert  ^{2}\right]\leq c_{8}n\left(m_{f,2}^{2}+m_{f,1}^{2}m_{f,2}+\mathcal{R}\right)\exp\left(c_{7}\left(n m_{f,2}+n^{2}m_{f,1}^{2}\right)\right).
\]
\subsection{Proof of Proposition~\ref{compl_gen_levy}}
We aim to minimize
$$\sum_{l=0}^L N_l\cdot(\delta_l^{-\alpha} + \Delta_l^{-1})$$
subject to
$$
\min(\Delta_L,\, \delta_L^{3-\alpha}) \le \varepsilon, \quad \sum_{l=0}^L\frac{\Delta_l}{N_l}\le\varepsilon^2.
$$
We denote
$$a_l = (\delta_l^{-\alpha} + \Delta_l^{-1}),\quad l=0,\ldots, L.$$

From Lagrange principle we get
\begin{gather*}
a_l = -\lambda\cdot N_l^{-2} \cdot \Delta_l \ \Rightarrow\  
N_l = \sqrt{(-\lambda)\cdot \Delta_l \cdot a_l^{-1}}\ \Rightarrow\ \\
\sum_{l=0}^L\frac{\Delta_l}{N_l} = \frac{1}{\sqrt{-\lambda}}\cdot \sum_{l=0}^L\frac{\Delta_l}{\sqrt{\Delta_l \cdot a_l^{-1}}} = \varepsilon^{2}\ \Rightarrow\ 
\sqrt{-\lambda} = \varepsilon^{-2}\cdot \sum_{l=0}^L \sqrt{\Delta_l\cdot a_l}\ \Rightarrow\\
N_l =  \sqrt{ \Delta_l \cdot a_l^{-1}}\cdot\varepsilon^{-2}\cdot \sum_{l=0}^L \sqrt{\Delta_l\cdot a_l}
\end{gather*}
So the cost has the representation
\begin{align*}
\sum_{l=0}^L N_l\cdot a_l &= \sum_{l=0}^L a_l\cdot\sqrt{ \Delta_l \cdot a_l^{-1}}\cdot\varepsilon^{-2}\cdot \sum_{k=0}^L \sqrt{\Delta_k\cdot a_k}\\
&=\varepsilon^{-2}\cdot\left( \sum_{l=0}^L\sqrt{\Delta_l\cdot a_l}\right)^2\\
&=\varepsilon^{-2}\cdot\left( \sum_{l=0}^L\sqrt{1 + \frac{\Delta_l}{\delta_l^{\alpha}}}\right)^2
\end{align*}
According to the restrictions on the bias we have
$$\Delta_L = M^{-L} = \varepsilon,\ \delta_L = \varepsilon^{\frac{1}{3-\alpha}} = \Delta_L^{\frac{1}{3-\alpha}}.$$
We now consider two cases.
\begin{enumerate}
\item We set $\delta_l = \varepsilon^{\frac{1}{3-\alpha}}$ constant on all the levels. Then the cost is bounded from above by
$$\sum_{l=0}^L N_l\cdot(\delta_l^{-\alpha} + \Delta_l^{-1})\preceq \varepsilon^{-2}\cdot \varepsilon^{-\frac{\alpha}{3-\alpha}}=\varepsilon^{-\frac{6- \alpha}{3-\alpha}}$$
\item In the second case we will set $$\delta_l = \Delta_l^{\frac{1}{2-\alpha}}.$$
Note, that $\delta_L = \Delta_L^{\frac{1}{2-\alpha}} < \Delta_L^{\frac{1}{3-\alpha}}$, so the bias condition is fulfilled.
Then the overall cost
\begin{align*}
\varepsilon^{-2}\cdot\left( \sum_{l=0}^L\sqrt{1 + \frac{\Delta_l}{\delta_l^{\alpha}}}\right)^2 &\asymp \varepsilon^{-2}\cdot\left( \sum_{l=0}^L\sqrt{1 + \Delta_l^{1 - \frac{\alpha}{2-\alpha}}}\right)^2 \\
&\asymp \varepsilon^{-2}\cdot\left( \sum_{l=0}^L\sqrt{1 + \Delta_l^{\frac{2-2\cdot\alpha}{2-\alpha}}}\right)^2. 
\end{align*}
\end{enumerate}
Combining all the cases together we get the statement.
\section{Appendix}

\begin{lem}
\label{gronwall}
Let \((y_n)\) and \((g_n)\) be two nonnegative sequences and let \(c\) be a nonnegative constant. If
\begin{eqnarray*}
y_n\leq c+\sum_{k=1}^n g_k y_k,\quad n\geq 0, 
\end{eqnarray*}
then 
\begin{eqnarray*}
y_n\leq c\exp\left(\sum_{k=1}^ng_k\right).
\end{eqnarray*}
\end{lem}
\begin{proof}
We have
\begin{eqnarray*}
y_n&\leq & c+\sum_{0\leq k <n } c g_k\prod_{k<j<n} (1+g_j) 
\\
&=& c+c\sum_{0\leq k <n }\left [\prod_{k\leq j<n} (1+g_j)-\prod_{k+1\leq j<n} (1+g_j)  \right]
\\
&=& c+c\left [\prod_{0\leq j<n} (1+g_j)-\prod_{n+1\leq j<n} (1+g_j)  \right]
\\
&=& c \prod_{0\leq j<n} (1+g_j)
\\
&\leq & c \exp\left(\sum_{0\leq j<n} g_j\right).
\end{eqnarray*}

\end{proof}
\bibliographystyle{plain}
\bibliography{ml__weak_bibliography} 

\end{document}